\documentclass[runningheads]{llncs}

\usepackage{amssymb}
\usepackage {amsmath}
\usepackage {color}
\usepackage{graphicx}
\usepackage{a4wide}

\newtheorem{observation}{Observation}

\newcommand{\ShoLong}[2]{#2} 

\newcommand{\rodrigo}[2][says]{** \textsc{rodrigo #1:} \textcolor{blue}{\textsl{#2}}**}

\newcommand{\matias}[2][says]{** \textsc{matias #1:} \textcolor{blue}{\textsl{#2}}**}

\newcommand{\new}[1]{#1}

\title{Stabbing segments with rectilinear objects\thanks{A preliminary version of this article appeared in \emph{Proc. 20th International Symposium on Fundamentals of Computation Theory}~\cite{fct}.}}

\author{
Merc\`e Claverol\inst{1}
\and
Delia Garijo\inst{2}
\and
Matias Korman\inst{3}
\and
Carlos Seara\inst{1}
\and \\
Rodrigo I. Silveira\inst{1}\thanks{Corresponding author.}
}

\institute{Universitat Polit\`ecnica de Catalunya, Spain. \\{\tt \{merce.claverol,carlos.seara,rodrigo.silveira\}@upc.edu}.
\and
Universidad de Sevilla, Spain. {\tt dgarijo@us.es}. \and
Tohoku University, Japan. {\tt mati@dais.is.tohoku.ac.jp}.
}

\let\doendproof\endproof
\renewcommand\endproof{~\hfill\qed\doendproof}

\begin{document}

\graphicspath{ {figures/} }

\maketitle

\begin{abstract}
Given a set $S$ of $n$ line segments in the plane, we say that a region
$\mathcal{R}\subseteq \mathbb{R}^2$ is a {\em stabber} for $S$ if $\mathcal{R}$ contains
exactly one endpoint of each segment of $S$. In this paper we
provide optimal or near-optimal algorithms for reporting all combinatorially different
stabbers for several shapes of stabbers. Specifically, we
consider the case in which the stabber can be described as the
intersection of axis-parallel halfplanes (thus the stabbers are
halfplanes, strips, quadrants, $3$-sided rectangles, or rectangles).
The running times are $O(n)$
(for the halfplane case), $O(n\log n)$  (for strips, quadrants, and
3-sided rectangles), and $O(n^2 \log n)$ (for rectangles).
\end{abstract}

\section{Introduction}

Stabbing or transversal problems have been widely investigated  in computational geometry and related areas.
 The general idea is to find a region with certain characteristics that intersects (or \emph{stabs}) a collection of geometric objects in a particular way.
This family of problems finds applications in multiple areas, such as 
automatic map generation~\cite{mg-06}, 
line simplification~\cite{ghms-apsmlp-93},
regression analysis~\cite{adkmpsy-ct-14,aegmps-dlcdc-11},
and even in bioinformatics, concretely in simplification of molecule chains for visualization, matching and efficient searching in molecule and protein databases \cite{DL}. 

A substantial amount of research in this area has focused on studying  the problem of stabbing a collection of line segments.
Several different criteria have been used to define what it means for a region to stab a set of segments:
(i) the region must contain exactly one endpoint of each segment,
(ii) the region must contain at least one endpoint of each segment, or
(iii) the region must intersect all segments (but no
restriction on the endpoints is given).


Most previous work on stabbers focuses on criteria (ii) or (iii), but in this paper we deal with criterion (i): we say that a region $\mathcal{R}\subseteq \mathbb{R}^2$ \emph{stabs} a set of segments $S$ if $\mathcal{R}$ contains exactly one endpoint of each segment of $S$; see Figure~\ref{no-stab-rectangle}(a). Concretely, we study the problem of computing all the combinatorially different stabbers of a given set of segments for some  types of stabbing regions. Two stabbing regions are considered \emph{combinatorially different} if and only if the sets of segment endpoints contained in each of them are different, otherwise they are said to be \emph{combinatorially equivalent}. Note that under our definition of stabbing region, the segments can be seen as pair of points (the interior of the segment does not play any role); however, as we will see afterwards, it will be convenient to keep referring to them as segments.


Our interest in criterion (i) mainly comes from its differences with respect to the other two criteria. The existence of stabbers in our setting is not always guaranteed, whereas for criteria (ii)-(iii) it is usually easy to find some stabbing region (there are always trivial solutions); this makes our  decision question interesting and non-trivial (see Fig.~\ref{no-stab-rectangle}). Also, the study of stabbers that require regions to contain exactly one endpoint of each segment fits into the general framework of \emph{classification} or \emph{separability problems} because the stabbers classify endpoints (the ones inside the region and those outside). This is an interesting connection because classification problems arise in many diverse applications, see for instance~\cite{KK}.

\subsection{Previous work}
We begin discussing previous work on the stabbing model studied in this paper, i.e., criterion (i), and only later we review work on criteria (ii) and (iii).

Perhaps the simplest stabber one can consider is a halfplane,
whose boundary is defined by a line. Hence, a stabbing halfplane is
defined by its boundary: a line that intersects all segments (note that the complement of a stabbing halfplane is another stabbing halfplane, with the same boundary line). In this
context, Edelsbrunner et al.~\cite{emprww} presented an
$O(n\log n)$ time algorithm for solving the
problem of constructing a representation of all combinatorially different stabbing lines
(with any orientation) of a given set of $n$ segments. Moreover, they also gave an $\Omega (n\log n)$ lower bound for the problem. However, the lower bound from~\cite{emprww} does not apply to the
decision problem (i.e., determining whether or not there exists a
line stabber for a set of segments). Afterwards, Avis et al.~\cite{ARW} gave an $\Omega(n\log n)$ lower bound that holds even for the decision problem in the fixed order algebraic decision tree model.

When no stabbing halfplane exists, it is natural to ask for a
stabbing \emph{wedge} (the stabbing region defined by the
intersection of two halfplanes).
Claverol et al.~\cite{mc,cggms}
studied the problem of reporting all combinatorially different
stabbing wedges of $S$. The time and space complexities of their
algorithm depend on two parameters of $S$, which in the worst case result in $O(n^3\log n)$ time and $O(n^2)$ space.
The authors of~\cite{cggms} also studied some other stabbers
such as double-wedges and zigzags.
The problem of computing stabbing \emph{circles} of a set $S$ of $n$ line segments in the plane has been studied very recently by Claverol et al.~\cite{CKPSS15},
obtaining the following results:
(i) a representation of all the combinatorially different stabbing circles for $S$
can be computed in $O(n^2)$ time and space;
and (ii) one can report all stabbing circles for a set of parallel segments in $O(n\log^2 n)$ time and $O(n)$ space.

Many of the problems studied for criteria (i) (and (ii)) can be formulated in terms of \emph{color-spanning} objects. In this case, the input is a set of $n$ colored points, with $c$ colors, and the goal is to find an object (rectangle, circle, etc.) that contains \emph{at least} (or {\em exactly}) $k$ points of each color. Our setting, with criterion (i), is the particular case in which $c=n/2$ and we want to contain {\em exactly} one point of each color class.
Barba {\em et al.}~\cite{similar} considered several problems related to color-spanning objects for criterion (i).
In particular, they present algorithms that can compute disks, squares, and axis-aligned rectangles that contain exactly one element of each of the $c$ color classes in $O(n^2c)$ time.

In a very recent paper, Arkin et al.~\cite{aghmz-mlscp-15} studied the computation of \emph{minimum length separating cycles} for a set of pairs of points. The problem can be seen as that of finding a polygonal stabber of minimum perimeter for a set of line segments. Since the problem is NP-hard, Arkin et al.~\cite{aghmz-mlscp-15} focused on approximation algorithms for several special cases.

%
%
%
\begin{figure}[tb]
\begin{center}
\includegraphics{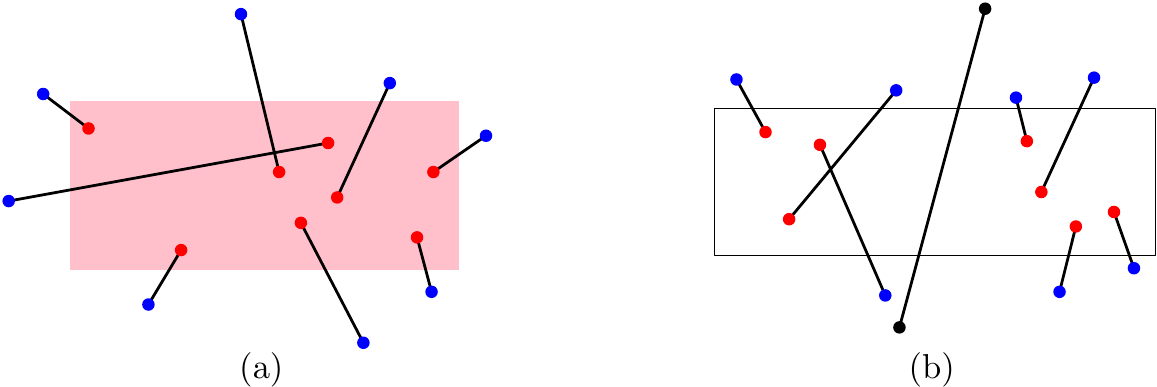}
\end{center}
\caption{(a) A set of segments that has a stabbing rectangle. (b)
A set of segments for which no stabbing axis-aligned rectangle exists.}\label{no-stab-rectangle}
\end{figure}

\subsubsection{Other stabbing criteria} 
There is plenty of related work for the two other criteria mentioned before: (ii) containing at least one endpoint, or (iii) intersecting each segment. 
In the following we briefly mention some of the most relevant papers for these criteria.


Most previous work for criteria (ii) and (iii) has focused on objects in two dimensions.
Atallah and Bajaj~\cite{AB} considered the problem of determining a stabbing line for a collection of objects in the plane.
Bhattacharya et al.~\cite{BCE} studied the computation of the shortest stabbing segment of a collection of segments (and lines) with criterion (iii).
Arkin et al.~\cite{adkmpsy-ct-14} considered the problem of computing convex stabbers for sets of segments in the plane, also for criterion (iii).
Several types of stabbers have been studied for criterion (ii) in the context of color-spanning objects, namely strips, axis-parallel rectangles~\cite{AHIKLMPS01,DGN}, and circles~\cite{ahikl-fcvd-06}; all of them can be computed in roughly $O(n^2\log c)$ time, for $c$ the number of different colors. 
Several papers~\cite{OR1,GS,R,BKM,MKGB,MGR,DKPPSS} have studied problem variants with the goal of optimizing perimeter or area of the convex polygon stabbing a set of segments, using criteria (ii) and (iii).

Some three-dimensional variants have also been studied, most notably the problem of computing stabbers or transversals of a set of segments in 3D~\cite{aw88,belzw05,fhph12}  and  a few stabbing problems for convex polyhedra~\cite{aw88,krs10,p93} .



Very recently, Arkin et al.~\cite{abcckm-cfc-15} studied other related problems in this setting, mainly focused on minimizing the number objects (e.g., unit squares) needed to contain exactly one point from each color.

Finally, we mention that there are other families of related problems about intersections and transversals, although usually formulated in much more general settings, in work about geometric transversal theory. We refer to~\cite{gpw-gtt-93} for a survey on the topic.

\paragraph{Contributions}

\begin{table}[t]
 \centering
    \begin{tabular}{ | l | l | p{5cm} |}
    \hline
   Stabber & Running time & Reference\\ \hline         \hline
    Horizontal line & $O(n)$ & This paper\\ \hline
    Line & $O(n \log n)$  & Edelsbrunner et al.~\cite{emprww}\\ \hline
    Horizontal strip & $O(n \log n)$ & This paper\\ \hline
 	Quadrant & $O(n \log n)$ & This paper\\ \hline
 	3-sided (axis-parallel) rectangle& $O(n \log n)$ & This paper\\
 \hline
 	Axis-parallel rectangle & $O(n^2 \log n)$  & This paper\\ \hline
 	Wedge & $O(n^3 \log n)$ & Claverol et al.~\cite{cggms}\\ \hline
 	Double-wedge & $O(n^4 \log n)$ & Claverol et al.~\cite{cggms}\\ \hline
	Zigzag  & $O(n^3 \log n)$ & Claverol et al.~\cite{cggms}\\ \hline
	Circle & $O(n^2)$ & Claverol et al.~\cite{CKPSS15}\\ \hline
    \end{tabular}
    \vspace{0.3cm}
    \caption{Worst-case running times of previous and new algorithms that, given a set of segments, report all combinatorially different stabbers of different shapes, containing exactly one endpoint from each segment.
We omit previous results that only work for certain classes of segments (e.g., parallel segments).
}
    \label{tab:results}
\end{table}

In this paper we consider stabbers
that can be described as the intersection of axis-parallel (i.e., horizontal or vertical) halfplanes.
Thus, the shapes we consider are halfplanes, strips, quadrants,
$3$-sided rectangles, and rectangles. This scenario, more restrictive than the ones previously studied for this definition of stabber~\cite{cggms,emprww,mc}, allows us to exploit the geometric structure of the possible stabbers to obtain much faster algorithms.

In Section~\ref{sec_2} we study the case in which the stabber is formed by one or two halfplanes (that is, halfplanes, strips and wedges). For that purpose, we introduce a general approach that partitions the plane into three regions: a \emph{red} region that must be contained in any stabber, a \emph{blue} region that must be avoided by any stabber, and a \emph{gray} region for which we do not have enough information yet.
\ShoLong{}{The algorithms are based on iteratively classifying segments and
updating the boundaries of these regions, a process that we call \emph{cascading}.} In Section~\ref{sec_3}
we extend this approach to $3$-sided rectangles.
To analyze the running time of our algorithm, we show that the maximum number of combinatorially different stabbers of this type is only $\Theta(n)$.
All algorithms in Sections~\ref{sec_2} and~\ref{sec_3} are asymptotically optimal.
Finally, in Section~\ref{sec_4}, we use the algorithm for 3-sided rectangles to find all different stabbing axis-parallel rectangles.
Table~\ref{tab:results} summarizes the results presented in this paper, together with the previous results for criterion (i) that put no constraints on the input segments.

To the best of our knowledge, the stabbing problems studied in this paper had not been considered before, except for rectangles, which are a particular case of the more general problem studied in~\cite{similar}.
The algorithm presented in Section~\ref{sec_4} (for axis-aligned rectangular stabbers) improves the result of~\cite{similar} for the particular case in which $c=n/2$ and one looks for an axis-aligned color spanning rectangle. Our algorithm is almost a linear factor faster, and also allows to report all possible solutions (whereas their algorithm only reports one).


\subsection{Preliminaries}\label{sub_sec_1.1}

The input to our problems consists of a set $S$ of $n$
segments in $\mathbb{R}^2$. For simplicity in the exposition, we assume that there
is no horizontal or vertical segment in $S$, and that all segments
have non-zero length. The modifications needed to make our
algorithms handle these special cases are straightforward, albeit
rather tedious.

Let $s=\{p,q\}$ be a segment of endpoints $p,q$; recall that in this paper a segment can be seen as a pair of points, and so we will say indistinctly \emph{endpoints of segments of $S$} and \emph{points of $S$}. For distinguishing the upper and the lower endpoints of $s$, say $p$ and $q$ respectively, we use $s=(p,q)$. Given a point $p\in \mathbb{R}^2$, we write $x(p)$ (resp. $y(p)$) for its $x$- (resp.
$y$-) coordinate. Let $y_b=\max_{(p,q)\in S}\{y({q})\}$ and
$y_t=\min_{(p,q)\in S}\{y(p)\}$; these values correspond to the
$y$-coordinates of the highest bottom endpoint and the lowest top
endpoint, respectively, of the segments of $S$. Let $s_b=(p_b,q_b)$ and $s_t=(p_t,q_t)$ where $y({q}_b) = y_b$ and $y(p_t)=y_t$.

We say that a point is \emph{red} if it is contained in the stabbing region $\mathcal{R}$, or \emph{blue} otherwise. Note that we consider
$\mathcal{R}$ a closed region, thus points on its boundary are considered red.
In addition, throughout this paper we assume that all solutions (i.e., red stabbing regions) are represented by a minimal number of red points on the region boundary: one point for halfplanes, two points for strips and quadrants, three points for 3-rectangles, and four points for rectangles.

We say that a stabbing region $\mathcal{R}$ is \emph{trivial} if there is a stabbing region $\mathcal{R}'$ described as the intersection of fewer halfplanes that is combinatorially equivalent to $\mathcal{R}$ (i.e., gives the same classification of the endpoints of the segments in $S$). In principle, our algorithms report all combinatorially different solutions, trivial and non-trivial. However, we note that it is very easy to detect when a solution is trivial (and thus, if needed we can avoid reporting those). Details of this are given in the proof of Theorem~\ref{theo_enumstrip} for strips, but the same approach extends to all the other stabbers considered.

%
%
%

Finally, we note that even though we present our algorithms for stabbers
defined by axis-parallel halfplanes, they extend to any two fixed orientations by making an appropriate affine transformation.

%
%
%

\section{Stabbing with one or two halfplanes}\label{sec_2}

In this section we look for stabbers that can be described as the
intersection of at most two halfplanes. That is, our goal is to
find halfplanes, strips, or quadrants that contain exactly one
endpoint from each segment.
Recall that such stabbing objects do not always exist.

\subsection{Stabbing halfplane}\label{sec:line}

As a  warm-up, we give a simple algorithm for determining if a
stabbing horizontal halfplane exists. That is, a horizontal line
such that one of the (closed) halfplanes defined by the line
contains exactly one endpoint from each segment. 
Thus, we
are effectively looking for a horizontal line that intersects all the
segments (such a line can define up to two complementary stabbers).

Although the algorithm is straightforward, we explain it for
completeness.
As we are dealing with horizontal stabbers, the problem becomes
essentially one-dimensional, and it will ease our presentation to
state the problem in that way.

All segments can be projected onto
the $y$-axis, becoming intervals. Considering the set $\overline{S}=\{\overline{s} \, | \, s\in S\}$
of projected segments, the question is simply whether all the intervals in
$\overline{S}$ have a point in common. Then, any horizontal line $y:=u$ stabbing
$S$ must have its $y$-coordinate between the two values $y_b$ and $y_t$, namely
$y_b \leq u \leq y_t$. Therefore such a line exists if and only if
$y_b\le y_t$. Moreover, whenever this condition happens, both the
upper and lower halfplanes will be stabbing halfplanes (and any other
horizontal halfplane will be combinatorially equivalent to one of the two).
This simple observation directly leads to a straightforward
linear-time algorithm.

\begin{observation}\label{algo_line}
We can find all the two axis-aligned stabbing halfplanes of a set $S$ of $n$ segments (or determine that none exists) in $O(n)$ time.
\end{observation}

\subsection{Stabbing strips}\label{sub_sec_2.2}

We now consider the case in which the stabbing region is a
horizontal strip. Note that the existence of stabbing halfplanes directly implies the existence
of stabbing strips, but the reverse is not true.

As in Section~\ref{sec:line}, we can ignore the $x$-coordinates of the endpoints,
project the points onto the $y$-axis, and work with the set
$\overline{S}$ instead. The endpoints of the classified segments
can be seen in the projection onto the $y$-axis as a set of blue
and red points.
It follows that there
is a separating horizontal strip for them if and only if the red
points appear contiguously on the $y$-axis. More precisely, the
points must appear on the $y$-axis in three contiguous groups,
from top to bottom, first a blue group, then a red group, and then
another blue group. We refer to the two groups of blue points
as the \emph{top} and \emph{bottom} blue points, respectively. We
denote the intervals of the $y$-axis spanned by them by $B_t$ and
$B_b$, respectively. The interval of the $y$-axis spanned by the
group of red points is denoted by $R$.
Moreover, since all points above $B_t$ and below $B_b$ must be
necessarily blue, we extend $B_t$ and $B_b$ from $+\infty$ and
until $-\infty$, respectively.
In this way the $y$-axis is partitioned into two blue intervals, one red interval, and two intervals separating them that we will consider of color gray.
See Fig.~\ref{fig:strip}(a).

\begin{figure}[tb]
\begin{center}
\includegraphics{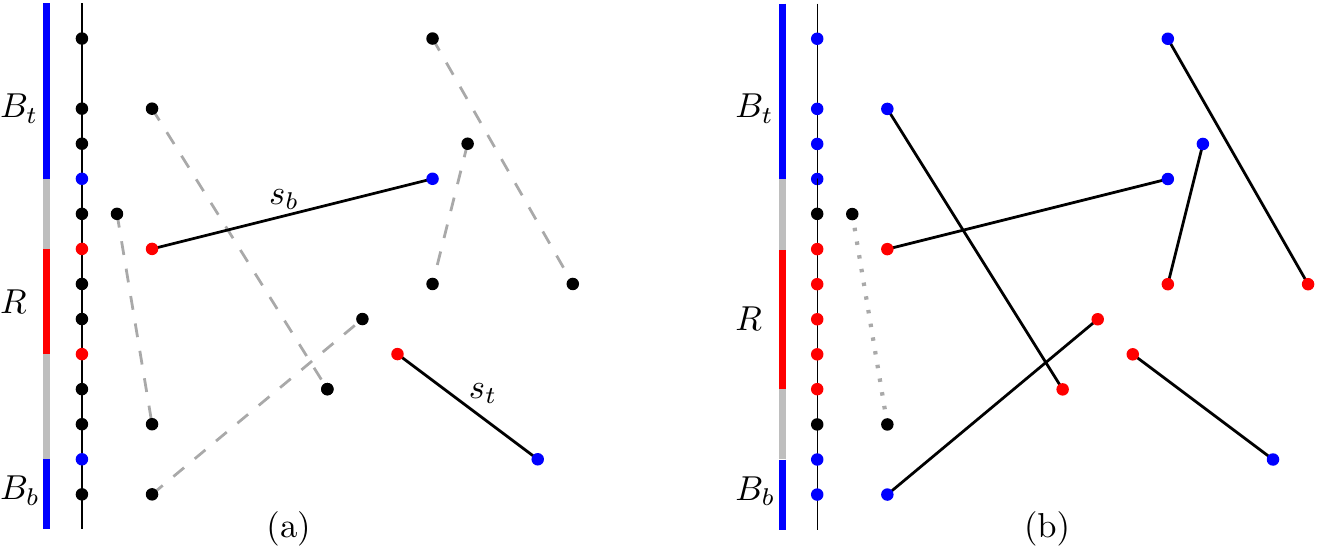}
\caption{Computing a horizontal strip stabber. (a) Result after
classifying $s_t$ and $s_b$. (b) Result after cascading; the
red region has grown, and only one segment remains unclassified.
In both figures, segments of $U$ are shown dotted, those of $W$
are dashed, and those of $C$ are depicted with a solid
line.}\label{fig:strip}
\end{center}
\end{figure}

Our algorithm relies on the following two observations.

\begin{observation}\label{obs_clasi}
Let $s_i=(p_i,q_i)$ and $s_j=(p_j,q_j)$ be two segments of $S$ such that the projected segment
$\overline{s}_i$ is above the projected segment $\overline{s}_j$ (in particular, this implies that $\overline{s}_i$ and $\overline{s}_j$ are disjoint). Then, any
horizontal stabbing strip for $S$ must contain $q_i$ and $p_j$.
\end{observation}

\begin{proof}
It follows from the fact that any strip that contains $q_i$ (and some endpoint of $s_j$) will also contain $p_i$. Similarly, strips that contain $q_j$ (and endpoints of $s_i$) will contain $q_i$. Since a strip cannot contain both endpoints of a segment of $S$, the claim follows.
\end{proof}

\begin{observation}\label{obs_init}
Any non-trivial horizontal stabbing strip $\mathcal{R}$ for $S$ contains points $q_b$ and $p_t$.
\end{observation}

\newcommand{\pfinit}{\begin{proof}
 If $y_b > y_t$, we can use Observation~\ref{obs_clasi} to prove our statement. Thus, it remains to consider the case in which $y_b < y_t$.

First observe that $\mathcal{R}$ must contain line $y=y_b$ or $y=y_t$. Indeed, if it contains neither of the two lines, the stabber is either in $(i)$ the halfplane $y>y_t$, $(ii)$ the halfplane $y<y_b$, or $(iii)$ in the strip $y_b< y< y_t$. However, in case $(i)$ $\mathcal{R}$ does not contain an endpoint of $s_t$. Similarly, any stabber in $(ii)$ cannot contain endpoints of $s_b$, and stabbers in $(iii)$ cannot contain endpoints of neither segment, a contradiction.

Without loss of generality, we can assume that $\mathcal{R}$ contains the line $y=y_t$ (and thus, it contains $p_t$). Assume, for the sake of contradiction, that $\mathcal{R}$ does not contain $q_b$. By definition of $s_b$, no bottom endpoint of a segment of $S$ can be in $\mathcal{R}$ (since they are below $s_b$). In particular, $\mathcal{R}$ must contain all upper endpoints of the segments of $S$, and thus $\mathcal{R}$ is equivalent to the halfplane $y\geq y_t$. However, this contradicts with the fact that $\mathcal{R}$ is a non-trivial stabber.
\end{proof}
}
\ShoLong{}{\pfinit}


With the previous observations in place, we can present our
algorithm. First we give an intuitive idea in rather general terms because it will
also be used in the upcoming sections.

The algorithm starts by classifying a few segments of $S$ using some geometric observations (in the case of strips, Observation~\ref{obs_init}). As soon as some points are classified, the plane can be partitioned into three regions: the {\em red} region (a portion of the plane that must be contained by any stabber), the {\em blue} region (a portion of the plane that cannot be contained in any stabber) and the {\em gray} region (the complement of the two other regions: regions of the plane for which we still do not have enough information yet). If a segment of $S$ has an endpoint in the blue region it can be classified as blue  (similarly if an endpoint is in the red region). The other endpoint of the segment must have the opposite color. This coloring may enlarge either the red or blue regions, which may further allow to classify other segments of $S$. This results in an iterative process that we call \emph{cascading procedure}.
We continue in this fashion until either a contradiction is found (that is, the red and blue regions overlap), or all segments of $S$ have been classified (and hence a stabber has been found).


At any instant of time $S$ is partitioned into three sets $C$, $W$, and $U$. A segment is in $C$ if it has been already classified, in
$W$ if it is waiting for being classified (there is enough
information to classify it, but it has not been done yet), or in
$U$ if its classification is still unknown. The algorithm is
initialized with $C=W=\emptyset$, and $U=S$.

\paragraph{Regions}
In addition to the three sets of segments $C$, $W$, and $U$, the algorithm maintains
red and blue candidate regions that are guaranteed to be contained
or avoided in any solution, respectively. When we are looking for a horizontal strip, these regions will also be horizontal strips. Thus, it suffices to maintain the projection of the regions on the $y$-axis. The blue and red regions are represented by three intervals $B_t$, $R$, and $B_b$. We define the \emph{blue region} as $B_t \cup B_b$ and the \emph{red region} as $R$. The complement
of the union of the blue and red regions is called the \emph{gray
region}. Note that the gray region, like the blue one, consists of
two disjoint components. During the execution of the algorithm the
regions will become updated as new segments become classified.

The algorithm starts by computing $s_b$ and $s_t$. If a stabbing
halfplane exists (i.e., $y_b\leq y_t$), then the only two (trivial) stabbing strips can be reported as described above.
Otherwise, Observation~\ref{obs_init} indicates how $s_b$ and $s_t$ should be classified:  $q_b$ and $p_t$ as red, and $p_b$ and $q_t$ as blue. See Fig. \ref{fig:strip}(a). Thus, they can be moved from $U$ to $W$.

\paragraph{Cascading procedure}
The procedure iteratively classifies segments in $W$
based on the red and blue regions. This is an iterative process in
the sense that the classification of one segment can make the blue
or red region grow, making other segments move from $U$ to $W$.

In the iterative phase, the algorithm picks a segment $s\in W$ (if any), assigns the corresponding colors to its endpoints, and moves $s$ from $W$ to $C$. If a newly assigned endpoint lies outside its corresponding zone, the red or blue area must grow to contain that point. Note that after the red or
blue region grows, other segments can change from $U$ to~$W$. As mentioned above, this  process continues until either: (i) a contradiction is found (the red region is forced to overlap with the blue region), or (ii) set $W$ becomes empty. See Fig. \ref{fig:strip}(b) for an example.

\begin{lemma}\label{lemma_keystrip}
If the cascading procedure finishes without finding a
contradiction, each remaining segment in $U$ has an endpoint in
each of the connected components of the gray region.
\end{lemma}
\begin{proof}
By definition, when the cascading procedure finishes, all segments still in $U$ must have both endpoints in the gray region. 
Assume, for the sake of contradiction, that there exists a
segment $s=(p,q)\in U$ \new{with both endpoints in} the same
gray component, say, the lower one. Recall that, by construction, the red region $R$ contains the
interval $[y_b,y_t]$. In particular, we have $y(p)<y_t$, giving
a contradiction with the definition of $y_t$.
\end{proof}

\begin{corollary}\label{cor_solution}
A horizontal stabbing strip for $S$ exists if and only if the cascading procedure finishes
without finding a contradiction.
\end{corollary}
\begin{proof}
From the above observations, it is clear that if the cascading
procedure finds a contradiction, then there will be no
solution. Likewise, if no contradiction is found, we can find a
solution by  extending the bottom boundary of $B_t$ until it
covers the upper gray component and extending the bottom boundary
of $R$ until it covers the lower gray component (or vice versa).
\end{proof}

\begin{theorem}\label{theo_strip}
Determining whether there exists a stabbing strip for a set $S$ of $n$
segments can be done in $O(n\log n)$ time and $O(n)$ space.
\end{theorem}
\newcommand{\pftheostrip}{
\begin{proof}
The correctness follows from the previous observations, thus we
focus on the running time. In order to implement the cascading
procedure efficiently, we must be able to quickly find if there is
a segment that has an endpoint in $B_t$, $B_b$, or $R$. Recall that the
red and blue regions are simply intervals. Thus, the queries can be done by
maintaining a balanced binary tree (sorted by $y$-coordinate) with
the endpoints of all unclassified segments. In this way, a segment
with an endpoint in, say, $B_t$ can be found in $O(\log n)$ time
by querying the tree with the endpoints of $B_t$. Every time an
endpoint inside $B_t$ is found, it is removed from the tree, and
moved into $W$. When classifying a segment, the limits of
$B_t$, $B_b$, or $R$ must also be updated accordingly.
Since these intervals have constant complexity they can be
updated in constant time. Further observe that no interval is classified more
than once, which in particular implies that the total time for the cascading operation is bounded
by $O(n\log n)$.
\end{proof}
}
\ShoLong{}{\pftheostrip}

\paragraph{Reporting all horizontal stabbing strips}
The above algorithm can be modified to report all combinatorially
different horizontal stabbing strips without \new{increasing the asymptotic 
running time}.

We execute the algorithm as usual until the cascading procedure finishes. Let $U=\{s_1,\ldots,s_k\}$ and let $\tau$ be the index of the segment of $U$ whose upper endpoint is lowest (i.e., for any index $i$ such that $s_i=(p_i,q_i)\in U$, it holds that $y(p_i)\geq y(p_{\tau})$). Likewise, we define $\beta$ as the index whose lower endpoint is highest. Any
stabbing strip will either contain: (i) all points in the upper
component of the gray region, (ii) all points in the lower
component of the gray region, or (iii) both $p_\tau$ and $q_\beta$. The
first two can be reported in constant time, and for the third
case, it suffices to classify the two segments, cascade, and
repeat the previous steps.

\begin{theorem}\label{theo_enumstrip}
All combinatorially different stabbing strips for a
set $S$ of $n$ segments can be computed in $O(n\log n)$
time.
\end{theorem}
\newcommand{\pfEnumStrip}{
\begin{proof}
The only new operations that the algorithm introduces are
searching for the indices $\tau$ and $\beta$, and executing further
cascading procedures. Since we keep the elements of $U$ in \new{a balanced binary tree}, the former operation can be done in logarithmic time. As in
the decision version, each segment can only trigger one cascade
operation, hence the total running time is also bounded by
$O(n\log n)$.

Finally, a brief remark about detecting trivial stabbers. Recall that a stabber is trivial if there is a simpler one (i.e., defined by less halfplanes) that classifies points in the same manner. 
For example, for the case of strips, a stabbing strip is trivial if there exits a stabbing halfplane that gives the same classification. 
It follows that a stabber is trivial if and only if when extending one of its sides to infinity, the region that is ``swept'' is empty of points of $S$. 
This situation can be checked in $O(\log n)$ time using an orthogonal range searching data structure~(e.g., \cite{McC85}).
Since our stabbers are defined by up to four halfplanes, four queries are enough to determine if a given stabber is trivial. This additional verification process adds a cost of $O(k\log n)$, where $k$ is the number of stabbers found. Since we know that $k\in O(n)$, this process does not increase the asymptotic running time of our algorithms. 
\end{proof}

}
\ShoLong{}{\pfEnumStrip}

Next we show that this running time is asymptotically optimal for computing the narrowest strip, and thus also for reporting all of them, at least under the model assumed in this paper.

\begin{theorem}\label{thm:maxgap}
The problem of computing the narrowest stabbing strip for a
set $S$ of $n$ segments is $\Omega(n\log n)$ in the algebraic decision tree model.
\end{theorem}

\begin{proof}
We present a reduction from the \textsc{Maximum Gap} problem, which is well-known \new{to have an $\Omega(n \log n)$ lower bound} in the algebraic decision tree model.
The input to \textsc{Maximum Gap} is a set of $n$ real numbers $X = \{x_1,x_2,\dots,x_n \}$, which we can assume to be positive. The output is the maximum difference between two numbers in $X$ that are consecutive in sorted order. Next we show that any \textsc{Maximum Gap} instance can be solved with an algorithm that computes the \emph{vertical} stabbing strip of minimum width for a set of $n$ segments.
\begin{figure}[tb]
\begin{center}
\includegraphics{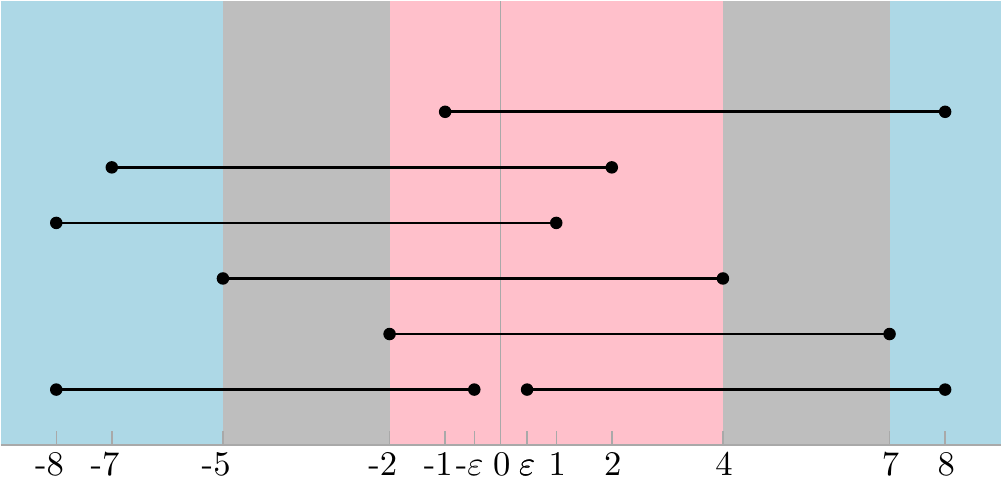}
\end{center}
\caption{Reduction from \textsc{Maximum Gap}. Segments $S$ for $X=\{7,4,1,2,8\}$. The width of the narrowest stabbing strip is $(C - \text{max gap($X$)})$, for $C$ a constant that depends on the maximum number in $X$.}
\label{fig:MaxGapReduction}
\end{figure}

Given the set of numbers $X$, we build a set of horizontal segments $S$ as follows. Let $x_M$ be the largest number in $X$.
For each $x_i \in X$, we add a segment $s_i$ whose endpoints have $x$-coordinates $-x_M+x_i-1$ and $x_i$. 
Since we are interested in vertical strips, the $y$-coordinates of the segments are irrelevant, so for simplicity the $y$-coordinate of the endpoints of $s_i$ is set to $i$, that is, $s_i=\{(-x_M+x_i-1,i),(x_i,i)\}$. Refer to Figure~\ref{fig:MaxGapReduction}.
Finally, to avoid having a trivial stabbing strip (i.e., a stabbing halfplane), we add two extra segments $\{(\varepsilon, 0),(x_M,0)\}$ and $\{(-x_M,0),(-\varepsilon,0)\}$ (for some arbitrarily small $\varepsilon>0$). 

Recall that a stabbing strip is given by the two red segment endpoints through which the (vertical) bounding lines of the strip go.
By Observation~\ref{obs_init}, any stabbing strip of $S$ must have its left defining line through a negative $x$-coordinate and its right defining line through a positive $x$-coordinate.
Consider a stabbing strip with right defining line $x=x_i$, and let $x_j$ be the $x$-coordinate of the next consecutive endpoint in $S$.
Then the strip contains all right  endpoints in $[0,x_i]$, and misses all those in $[x_j,x_M]$.
Since the right endpoint of $s_j$ is not contained in the strip, its left endpoint must be $(-x_M+x_j-1,j)$.
By construction, left  endpoints appear, from left to right, in the same order as right  endpoints, thus the strip cannot contain any left endpoint to the left of $(-x_M+x_j-1,j)$, because the next left endpoint corresponds to $s_i$, which is already stabbed.
Therefore its left defining line must go through  $(-x_M+x_j-1,j)$.

Thus, for each $x_i\in X$ there is exactly one vertical stabbing strip associated with it. The boundary of such strip is given by two lines of the form $x=x_i$ and $x=-x_M-1+x_{j}$, and in particular it has width $x_M+1-(x_{j}-x_{i})$. 
Note that $(x_{j}-x_{i})$ is the gap between $x_i$ and the next consecutive number in $X$, $x_j$.
It follows that a strip that has minimum width maximizes the gap between its associated endpoint and the next consecutive number in $X$, thus corresponds to the maximum gap in $X$.
\end{proof}

If we assume that each strip is described by the endpoints from $S$ on each boundary line, then one can identify the narrowest strip in $O(n)$ time, and thus can also find the maximum gap in $X$ within the same running time.

\begin{corollary}
The problem of computing all combinatorially different stabbing strips for a
set $S$ of $n$ segments, assuming that each strip is described by one point from $S$ on each boundary line, is $\Omega(n\log n)$ in the algebraic decision tree model.
\end{corollary}

Note that given that strips are particular cases of 3-rectangles and rectangles, the previous lower bound also applies to the problems in Sections~\ref{sec_3} and~\ref{sec_4}.

\subsection{Stabbing quadrants}\label{quadrants}


We now extend the cascading approach to stabbing quadrants (i.e., a wedge
formed by a horizontal and a vertical ray from a common point, the
\emph{apex} of the quadrant). There are four types of quadrants.
Without loss of generality, we concentrate on the bottom-right
type. Throughout this section, the term \emph{quadrant} refers to a bottom-right quadrant. Other types can be handled analogously.

For a segment $s=\{p,q\}$, let $Q(s)$ denote its
\emph{bottom-right quadrant}; that is, the quadrant with apex at
$(\max\{ x(p), x(q) \}, \min \{ y(p), y(q) \})$. See
Fig.~\ref{fig:quadrants}(a).

\begin{observation}\label{obs_initquad}
Any quadrant classifying a segment $s$ must contain $Q(s)$.
\end{observation}

Given the set $S$ of segments, the \emph{bottom-right quadrant} of $S$, denoted by $Q(S)$, is the (inclusion-wise) smallest quadrant that contains $\cup_{s \in S} Q(s)$; see Fig.~\ref{fig:quadrants}(b).
%
%


We now need the equivalent of Observation~\ref{obs_init} to create the initial partition of the plane into red, blue, and gray regions.

\begin{corollary}\label{cor:quadrant}
Any stabbing quadrant of $S$ must contain $Q(S)$.
\end{corollary}

%
%

\begin{figure}[tb]
\begin{center}
\includegraphics{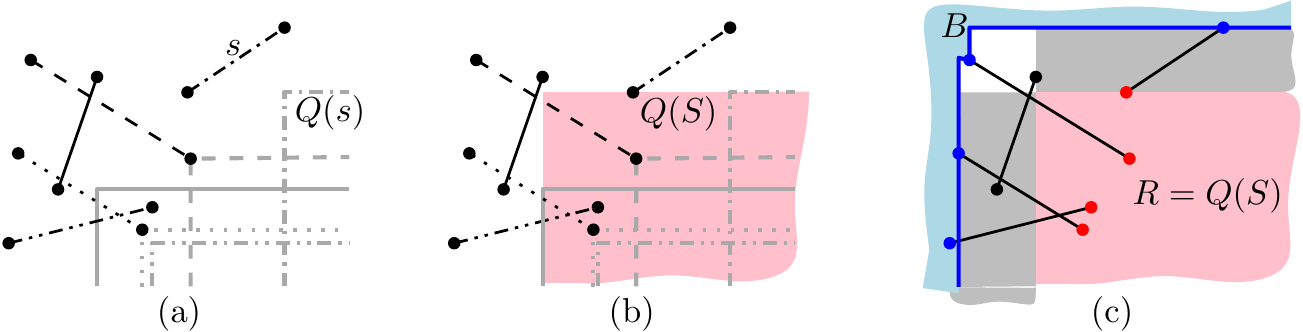}
\end{center}
\caption{(a) Five segments and their individual bottom-right
quadrants. (b) Bottom-right quadrant $Q(S)$ of the set of
segments $S$. (c) Initial classification given by $Q(S)$,
and the red, blue, and gray regions. Note that there is a region
(white in the figure) that cannot contain endpoints of $S$.}
\label{fig:quadrants}
\end{figure}

\paragraph{Regions}
The partition is as follows: the red region $R$ is always defined as the inclusionwise smallest quadrant that contains $Q(S)$ and all classified points. 
At any point during the execution of the algorithm, let $a=(x_R,y_R)$ denote
the apex of $R$. Any blue point $b$ of a classified segment
forbids the stabber to include $b$ or any point above and to the
left of $b$ (i.e., in the top-left quadrant of $b$). Moreover, if
$b$ satisfies $y(b)\leq y_R$ or $x(b)\geq x_R$ a whole halfplane
will be forbidden (or, equivalently, blue). Thus, the
union of \new{the boundary of} such blue regions forms a staircase polygonal line (see Fig.~\ref{fig:quadrants}(c)). Initially,
we take the blue region defined by a staircase with a single point at $(-\infty,\infty)$.
We say that a point $p=(x,y)$ is in the gray region if it is not
in a red or blue region, and satisfies either $x> x_R$ or
$y< y_R$ (see Fig.~\ref{fig:quadrants}(c)). As before,
observe that the gray region is the union of two connected
components (which we call {\em right} and {\em down} components).
Note that, in this case, there are regions in the plane (which we
call \emph{white}) that do not belong to either red, blue, or gray
regions. However, our first observation is that no endpoint of an
unclassified segment can lie in the white region.

\begin{observation}\label{obs_nowhite}
A segment $s\in S$ containing an endpoint in the white region must
contain its other endpoint in the red region.
\end{observation}
\begin{proof}
This fact follows from the definition of $Q(s)$: assume that there
exists a segment $s\in S$ that has an endpoint in the white region
and an endpoint in the right gray component. In particular, the
$y$-coordinate of both endpoints is larger than $y_R$. Thus, the
quadrant $Q(s)$ is not contained in $R$. However, recall that, by
construction, we have $Q(s)\subseteq Q(S)\subseteq R$; obtaining a
contradiction. Similar contradictions are found when the other
endpoint of $s$ is in the other gray component, or in the white or
blue regions.
\end{proof}

Initially we classify all segments that have one endpoint inside
$Q(S)$ in $W$ (and all the rest in $U$). As before, we
apply the cascading procedure until we find a contradiction (in which case we conclude that a
stabbing quadrant does not exist), or set $W$ eventually becomes empty. In this case we have a very strong characterization of the remaining unclassified segments.

\begin{lemma}\label{lemma_keyquadrant}
If the cascading procedure finishes without finding a
contradiction, each remaining segment in $U$ has an endpoint in
each of the gray components.
\end{lemma}
\begin{proof}
As with the strip case, the cascading procedure cannot finish if a
segment in $U$ has an endpoint in either a red or a blue region. As
argued in the proof of Observation~\ref{obs_nowhite}, if both
endpoints lie in the same gray component, quadrant $Q(s)$ is
not contained in $R$, giving a contradiction.
\end{proof}

Thus, if no contradiction is found we can extend $R$ until it contains
one of the two gray components to obtain a stabber. This can be done because, by construction, each of the connected
components of the gray region forms a bottomless rectangle (i.e, the intersection of three axis aligned halfplanes) that
shares a corner with the apex of $R$.

As in the strip case, we can use the following approach to report all combinatorially
different quadrants in an analogous fashion to Theorem~\ref{theo_strip}: any stabbing quadrant will either completely contain one of the two gray components, or it will contain at least a segment of each of the two components.

\begin{theorem}\label{theo_allquadrants}
All the combinatorially different stabbing quadrants for a
set $S$ of $n$ segments can be computed in $O(n\log n)$ time.
\end{theorem}
\begin{proof}
The proof of correctness is analogous to the strip case, thus we argue
about the running time. The red and gray regions have constant
complexity, and so they can be updated in constant time. The
boundary of the blue region can be determined by $\Omega(n)$
points. However, notice that the points in the staircase are
sorted in both the $x$-coordinates and the $y$-coordinates. Thus,
we can insert a point in the blue region in $O(\log n)$ time if
the points in the blue region are stored in a sorted fashion. Once
we have inserted a new point, we check if this increase makes the
red and blue regions intersect. Since we only need to compare
the upper left quadrant defined by the new point with the red
region, this can be done in constant time. Overall, we can update the
red, blue and gray regions in $O(\log n)$ time.

Now we need an efficient method to update the segments in $W$
whenever the red or blue region changes. For that purpose we can
use any data structure suitable for range searching with a
quadrant. For example, one based on priority search
trees~\cite{McC85} will suffice: given a bottomless rectangular
query region, after $O(n\log n)$ preprocessing time and using
$O(n)$ space, we can report  in $O(\log
n+k)$ time all segments containing an endpoint in
the query region (or report that the region is empty), where $k$ is the number of segments within that
region. Note that the structure is dynamic and can handle both
deletions and insertions in $O(\log n)$ time.

Whenever a point is colored, its corresponding region is increased with
a rectangular area (which may degenerate to a halfplane or bottomless
rectangle). Thus, we can find the segments that contain at least one
endpoint in the new segments with a rectangular range searching query. Each time we
report the segment, we insert it into $W$, delete it from $U$ (and from the range searching data structure, so it is not reported again). Thus, all
queries will be computed in total $O(n\log n)$ time and $O(n)$
space.

We can report all stabbing quadrants using the same approach as in the strip case: apply the cascading procedure until no more points remain in $W$, report the two stabbers that completely contain one of the gray components, or classify the segments in both components whose endpoints are closest to the red region, and further apply the cascading procedure.
\end{proof}

Similarly to the case of stabbing strips, we can show that our algorithm is asymptotically optimal if regions have to be represented by endpoints on the boundary.

\begin{theorem}
The problem of computing all combinatorially different stabbing quadrants for a
set $S$ of $n$ segments, assuming that each quadrant is described by one endpoint from $S$ on each boundary line,  is $\Omega(n\log n)$ in the algebraic decision tree model.
\end{theorem}
\begin{proof}
The result is obtained by using a similar construction to the one used for Theorem~\ref{thm:maxgap} (see Fig. \ref{fig:MaxGapReduction}), with two simple modifications.

The first one is making all horizontal segments collinear, placing all of them on the $x$-axis.
(This is done to simplify the arguments below, but larger distances between consecutive segments also work as long as they are small enough.)

The second one is rotating the coordinate system counterclockwise by $45^\circ$.
In this way, there is a 1-to-1 correspondence between stabbing strips in the construction of Theorem~\ref{thm:maxgap} and stabbing quadrants in the modified version.
However, since quadrants are unbounded in two directions, the concept of narrowest strip does not naturally extend to quadrants. Instead, we look at the width of the gray regions induced.

Recall that for quadrants the gray region has two connected components, and each one is described as the intersection of three quadrilaterals. More importantly, in the problem instance we construct, both regions will have the same width. 

Specifically, if the upper boundary of the red region has an endpoint of $s_i$, we know that the endpoint in the associated blue region must be from $s_j$, where $j$ is the index of the value that goes after $i$ in $X$. Since we rotated the endpoints $45^\circ$, the width of the gray region will be equal to $1/\sqrt{2}$ times the difference between $x_i$ and $x_j$. 
%
%
Therefore the stabbing quadrant maximizing the width of its gray regions gives the maximum gap in $X$.

Finally, note that if one is given all combinatorially different stabbing quadrants for $S$, assuming that each of them is described by one endpoint from $S$ on each boundary line, then the one with widest gray regions can be identified in $O(n)$ time. 
The reduction follows.
\end{proof}

\section{Stabbing with three halfplanes}\label{sec_3}

We now consider the case in which the stabber is defined as the intersection
of three halfplanes. As in the quadrant case, it suffices to consider those
of fixed orientation. Thus, throughout this section, a $3$-\emph{rectangle} refers to a rectangle that is missing the lower boundary edge, and that extends infinitely towards the negative $y$-axis (also called \emph{bottomless rectangle}).
As in the previous cases, we solve the problem by partitioning the plane
into red-blue-gray regions.
However, in this case the algorithm will need an additional sweeping phase.

\subsection{Number of different stabbing 3-rectangles}
First we analyze the number of combinatorially different stabbing 3-rectangles that a set $S$ of $n$ segments can have.
This analysis will lead to an efficient algorithm to compute them.

\begin{figure}[h]
\centering
\includegraphics[width=0.25\textwidth]{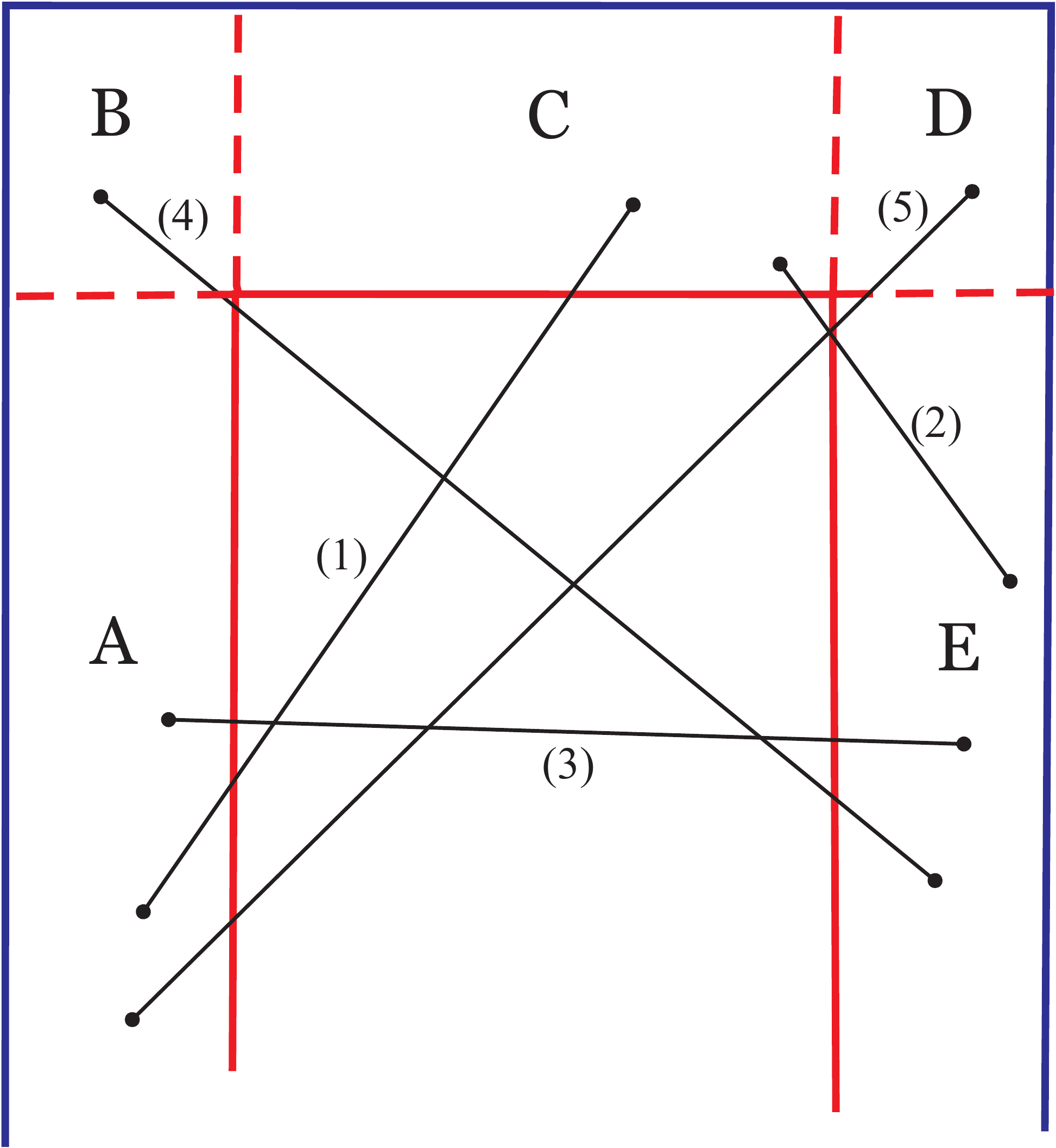}
\caption{Types of segments after cascading.}
\label{fig:SegmentTypes}
\end{figure}

We start by defining a region that must be included in any stabbing $3$-rectangle. Recall that $s_b$ is the segment of $S$ with highest bottom endpoint, and $q_b$ is its bottom endpoint. Analogously, we define $s_{r}$ and $s_\ell$ as the segments with leftmost right endpoint and rightmost left endpoint, respectively.
Let $p_{\ell}$ be the left endpoint of $s_{\ell}$ and
let $p_r$ be the right endpoint of $s_{r}$.
See Fig.~\ref{fig:3rect}(a).
Finally, we define lines $L_\ell$, $L_r$ as the vertical lines passing through $p_\ell$ and $p_r$, respectively.
Analogously, $L_b$ is the horizontal line passing through $q_b$.

\begin{figure}[tb]
\begin{center}
\includegraphics{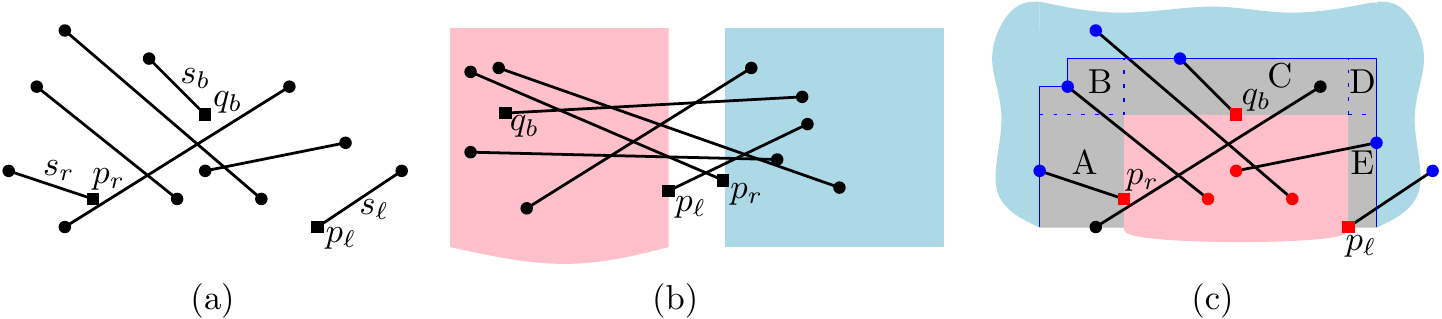}
\end{center}
\caption{(a) Example with points $q_b$, $p_{\ell}$ and $p_r$
highlighted as squares. (b) Any stabber that contains $p_\ell$ and avoids $p_r$ contains the left endpoints of all segments of $S$, and is combinatorially equivalent to the stabbing halfplane $x\leq x(p_\ell)$.
(c) Red and blue regions after the initial cascading procedure has finished, and partition of the gray region into subregions $A$, $B$, $C$, $D$, $E$.}
\label{fig:3rect}
\end{figure}

\begin{lemma}\label{lem_cascadinit3}
Any non-trivial stabbing $3$-rectangle $\mathcal{R}$ for $S$ must contain the intersection points of lines $L_{\ell}$, $L_r$ and $L_b$.
\end{lemma}
\begin{proof}
This proof is the analogous to that of Observation~\ref{obs_init} but for 3-rectangles. First we prove that $\mathcal{R}$ must contain some point of $L_b$. Indeed, recall that $\mathcal{R}$ is a $3$-rectangle, hence if it has empty intersection with $L_b$, then it must be contained in the lower halfplane defined by $L_b$. In particular, it cannot contain any endpoint of $s_b$, reaching a contradiction.

Now, we distinguish between two cases depending on the respective locations of $p_\ell$ and $p_r$. First consider the case in which $x(p_r)\leq  x(p_\ell)$. As in Observation~\ref{obs_init}, we can show that $\mathcal{R}$ must contain points of both $L_\ell$ and $L_r$. Otherwise, it would not contain an endpoint of either $s_\ell$ or $s_r$, respectively. That is, $\mathcal{R}$ must intersect the three lines $L_b$, $L_\ell$, and $L_r$. Moreover, since $\mathcal{R}$ is a 3-rectangle, it must contain the intersection points of the three lines.

It remains to consider the case in which $x(p_r)>x(p_\ell)$. Recall that in this case two trivial vertical halfplane stabbers exist (halfplanes $x\leq x(p_\ell)$ and $x \geq x(p_r)$). As in Observation~\ref{obs_init}, we show that $\mathcal{R}$ cannot avoid both $L_r$ and $L_\ell$: if it would have empty intersection with both lines, then it would be completely to the right, left, or in between these lines.
Hence it would avoid $s_r$, $s_\ell$ or both $s_r$ and $s_\ell$, a contradiction.
Thus, $\mathcal{R}$ must have nonempty intersection with either $L_r$ or $L_\ell$. We claim that if it intersects with exactly one of them, then $\mathcal{R}$ contains the same points of $S$ as one of the two trivial halfplane stabbers.
Indeed, assume for the sake of contradiction that $\mathcal{R}$ intersects with $L_\ell$ but has empty intersection with $L_r$. In this case, $\mathcal{R}$ cannot contain any point in the right halfplane of $L_r$, see Fig.~\ref{fig:3rect}(b).
By definition of $s_\ell$ and $s_r$, each segment of $S$ must contain an endpoint on or to the right of $L_r$ (none of which can be contained in $\mathcal{R}$), and another endpoint on or to the left of $L_\ell$. In particular, if $\mathcal{R}$ does not intersect with $L_r$, it must contain all endpoints on or to the left of $L_\ell$, and is combinatorially equivalent to the trivial halfplane stabber $x\leq x(p_{\ell})$. The case in which $\mathcal{R}$ intersects with $L_r$ and avoids $L_\ell$ is analogous.

Therefore, we conclude that  $\mathcal{R}$ must intersect at least one point from each of $L_b$, $L_\ell$, and $L_r$, and thus it must contain the intersection points of the three lines.
\end{proof}


The result above gives a way to initialize the red region (as the region determined by lines $L_b$, $L_\ell$, and $L_r$) to start the usual cascading procedure (whereas the blue region is initialized as empty). Segments with an endpoint in the red region are placed in $W$ (note that \new{it may be} that no such segment exists), and the usual cascading procedure is executed. If no contradiction is found, the classified segments define a red and a blue region that must be included or avoided in any solution, respectively.  Specifically, the red region is the inclusionwise smallest 3-rectangle that contains all points classified as red (this region must contain at least the intersection points of Lemma~\ref{lem_cascadinit3}). The blue region is the set of points whose inclusion would force a 3-rectangle to contain some point that has already been classified as blue. As usual, the area that is neither red nor blue is the \emph{gray region}. Once the cascading procedure has finished, all remaining unclassified segments must have both of their endpoints in the gray region.

It will be convenient to distinguish between different parts of the gray region, depending on their position with respect to the red region.
We differentiate between five regions, named A,B,C,D,E, as depicted in Fig.~\ref{fig:3rect}(c). These five regions are obtained by drawing horizontal and vertical lines through the two corners of the red region.

We say that the \emph{type} of a segment $s$ is XY, for $\text{X,Y}\in \{ A,B,C,D,E \}$ if $s$ has one endpoint in region X and the other endpoint in region Y. The next lemma shows that, after (a successful) cascading, there are only a few possible types.

\begin{lemma}\label{lem_segtypes}
Any unclassified segment after the cascading procedure is of type in $AC$, $AD$, $AE$, $BE$, or $CE$.
\end{lemma}
\begin{proof}
The claim follows from the fact that there cannot be a segment completely to the right, left or above the red region. More formally, assume that there exists a segment $s\in S$ of type $XY$, for $X,Y\in \{A,B\}$. By Lemma~\ref{lem_cascadinit3}, we know that $p_\ell$ is classified as red, and in particular both endpoints of $s$ are to the left of $p_\ell$. However, this contradicts with the definition of $s_\ell$. Similarly, segments of type $XY$, for $X,Y\in \{B,C,D\}$ would give a contradiction with $s_t$, and with $s_r$ for $X,Y\in \{D,E\}$.
\end{proof}

\begin{figure}[tb]
\centering
\includegraphics{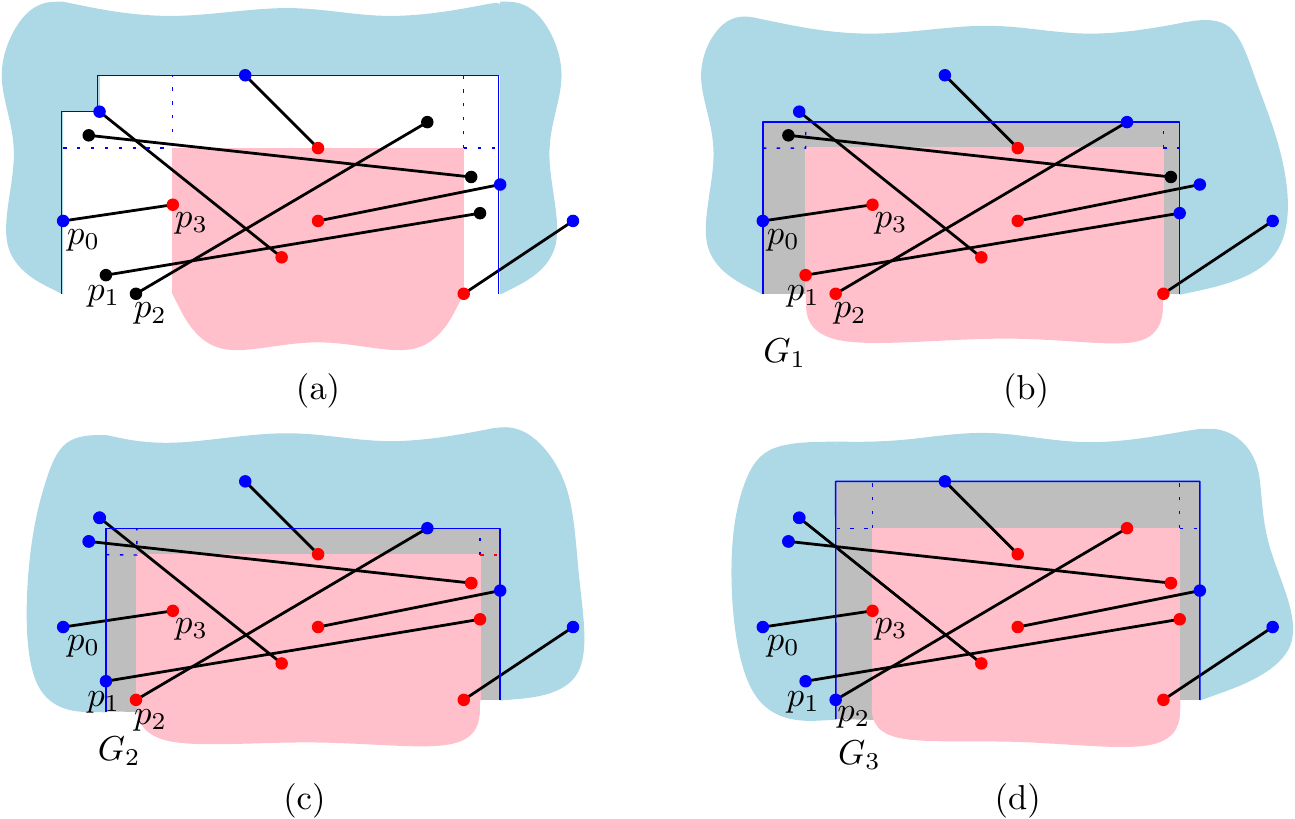}
\caption{(a) Initial situation; region A contains points $\{p_1,p_2\}$. (b)-(d): sequence of gray regions $G_1, G_2, G_3$.}
\label{fig:GiRegions}
\end{figure}

Now, consider the red and blue regions obtained after cascading.
Let $p_1,\dots,p_k$ be the endpoints of segments of $S$ inside region A, as they appear from left to right.
In addition, we define $p_0$ as the blue point defining the blue boundary of region A, and  $p_{k+1}$ as the red point defining the red boundary of region $A$.
See Fig.~\ref{fig:GiRegions}(a).
Let $G_i$, for $1 \leq i \leq k+1$, be the gray region obtained after classifying points $p_1,\dots,p_{i-1}$ as blue, $p_i,\dots,p_k$ as red, and performing a cascading procedure (see Fig.~\ref{fig:GiRegions}(b-d)). If any of those cascading operations results in a contradiction being found, we simply set the corresponding $G_i$ to be empty.

Observe that, if for some $i$,  $G_i\neq \emptyset$, then any unclassified segment (i.e., segment with both endpoints in $G_i$) must be of type BE or CE.


Next we upper-bound the number of ways in which $G_i$ can be completed to a solution.

\begin{lemma}\label{lem_combsol}
Let $G_i$ be defined as above for some $1 \leq i \leq k+1$ such that $G_i\neq \emptyset$, and let $n_i$ be the number
of unclassified segments in $G_i$. Then there are at most
$n_i$ combinatorially different solutions that use the classification induced by $G_i$.
\label{lem:OneRegion}
\end{lemma}
\begin{proof}
Since all unclassified segments are of type BE or CE, each of them has one endpoint in region E.
In any solution based on the classification induced by $G_i$, the right boundary of the red region must cross somewhere within E.
Without loss of generality, assume it goes through a point in E (since there cannot be endpoints
of unclassified segments in region D).

This implies that if we fix the right side of such a solution, the classification of all
unclassified segments becomes fixed. To see this, let $p$ be the endpoint in E through which the right side of
the solution goes. Fixing $p$ implies classifying all other segments that have an endpoint in E (those
to the left of $p$ must be red, the others blue).
Since all remaining unclassified segments have an
endpoint in E, they all become classified. It follows that there is at most one combinatorially
different solution for each point in E, of which there are $n_i$.
\end{proof}

We are now ready to show that the total number of combinatorially different solutions is linear.

\begin{lemma}\label{lem_atmostn}
For any set $S$ of $n$ segments there are $O(n)$ combinatorially different stabbing 3-rectangles.
\end{lemma}

\begin{proof}
It suffices to show that an unclassified segment in region $G_i$ cannot be unclassified in $G_j$ for $i\neq j$. 
This fact, combined with Lemma~\ref{lem:OneRegion}, implies that the total number of possible solutions is $O(n)$.
%
%
We can see $G_i$ as defined by three intervals, one on the left (i.e., region A), one on top (region C) and one on the right (region E). Actually, regions B and D can be more complex in $G_i$ (the upper left and right boundaries may form a increasing and decreasing staircase shape, respectively), but we ignore this fact for the proof.

The left interval is initially defined by the gap in $x$-coordinates between $p_{i-1}$ and $p_i$, although the posterior cascading procedure can limit it further. In particular, an unclassified segment in $G_i$ must have an endpoint whose $x$-coordinate lies between those of $p_i$ and $p_{i+1}$.

Assume for the sake of contradiction, that there exist two indices $i,j$ such that $1\leq i<j\leq k+1$ and a segment $s\in S$ such that $s$ is unclassified in both $G_i$ and $G_j$. We distinguish two cases depending on whether $s$ is of type BE or CE.

First consider the (simpler) case in which $s$ is of type BE. From the above reasoning we conclude that $s$ has an endpoint whose $x$-coordinate lies between those of $p_i$ and $p_{i+1}$, and another endpoint between $p_j$ and $p_{j+1}$. These intervals are disjoint whenever $i\neq j$. In particular, this implies that both endpoints of $s$ are to the left of the red region, which would contradict Lemma~\ref{lem_segtypes}.


It remains to consider the case in which $s$ is of type CE. We claim that either the top interval or the right interval of $G_i$ is disjoint from the corresponding interval of $G_{j}$.
Let $\{p_i,q_i\}$ be the segment with endpoint at $p_i$. Recall that this segment can be of type AC, AD or AE. Since we assumed that $i<j$, we know that $p_i$ is red, and $q_i$ is blue in $G_i$, whereas the opposite assignment occurs in $G_j$. This implies that the region where $q_i$ is (C, D or E), changes its gray interval to a new interval that must be disjoint from the previous one. Namely, if $q_i$ lies on C or D, then the red region changes from ending strictly below $q_i$ to at least containing $q_i$. On the other hand, if $q_i$ lies on E, the red region changes from ending strictly to the left of $q_i$ to containing it. Thus, either the top or right intervals of $G_i$ and $G_{j}$ are disjoint, implying no segment of type CE can be in both.
%
\end{proof}

\subsection{Algorithm}

The previous results give rise to a natural algorithm to generate all combinatorially different stabbing 3-rectangles.
The algorithm has two phases:

\begin{description}
\item[Initial cascading]
We initialize the red region using Lemma~\ref{lem_cascadinit3} and launch the usual cascading procedure. If this cascading finishes without finding a contradiction,
we obtain a red and a blue region that must be included and avoided by any 3-rectangle stabber, respectively.

\item[Plane sweep of region A] Now we sweep the points in region A from left to right. In the $i$th step of the
sweep, we classify points $p_1,\dots,p_{i-1}$ as blue, and points $p_i,\dots,p_k$ as red. After each such step, we
perform a cascading procedure. If the cascading gives no contradiction, we are left with a gray region
$G_i$ and a number of unclassified segments that must be of type BE or CE.
Then we sweep the endpoints of the unclassified segments in region E from left to right (we call this the {\em secondary} sweep). At each step
of the sweep, we fix those to the left of the sweep line as red, and those to the right of the sweep line as
blue, and perform a third cascading procedure. From the proof of Lemma~\ref{lem:OneRegion}, we know
that each step of this second sweep, after the corresponding cascading procedure, can produce at most one different solution.
\end{description}


\begin{theorem}\label{theo_bottomless}
All combinatorially different stabbing 3-rectangles for a
set $S$ of $n$ segments can be computed in $O(n \log n)$ time.
\end{theorem}
\begin{proof}
First we note that there is an $O(n\log n)$ time preprocessing cost for obtaining the
different orders of the endpoints of the segments by increasing and
decreasing of their $x$ and $y$ coordinates.

Next we analyze the time needed in the initial cascading process. Each time an endpoint is colored, we can check in at most $O(\log n)$ time if it creates a contradiction: this is done similarly to the previous section, where we maintain the intervals in regions A, C and E, and the staircases in regions B and D in dynamic data structures that allow for logarithmic-time updates. Thus, the total time needed to check whether there is an initial solution is $O(n\log n)$. Using the same structure, we can detect contradictions and update regions when a point is recolored (in either the sweep of region A or a secondary sweep).

In addition, after each step of the sweep in region A we must revert certain colorings.
Recall that in the $i$th step we classify points $p_1,\dots,p_{i-1}$  as blue and points $p_i,\dots,p_k$ as red.
After the i$th$ step, we need to undo the colorings that are a consequence of setting $p_i$ red, since from the next step on, $p_i$ will be blue.
This can be done efficiently by precomputing, for each point $p$, the largest $j$ such that a $p_j$ in region A will force a coloring of $p$.
In this way, the total time spent reverting colorings, after preprocessing, is $O(n)$.
The extra preprocessing time needed to precompute this information is $O(n \log n)$.

Overall, the sweep process of region A (disregarding the secondary sweeps) needs $O(n \log n)$ time, since at most $O(\log n)$ time is spent per each segment. The secondary sweep of $n_i$ points in region E, including all the cascading procedure, can be done in $O(n_i \log n_i)$ time. Since the groups of points swept in region E in each step are disjoint, the total cost of all cascadings in secondary sweeps is also bounded by $O(n \log n)$.
%
\end{proof}

\section{Stabbing rectangles}\label{sec_4}

In this section we consider the computation of rectangular stabbing regions.
Even though we would like to use a cascading-based approach, like the ones used in previous sections, it is not clear how this can be done.
Indeed, a fundamental property for the cascading procedure to work is that we can find some region of the plane that must be contained in all non-trivial stabbers (see Observation~\ref{obs_init}, Observation~\ref{obs_initquad}, and Lemma~\ref{lem_cascadinit3}).
Unfortunately, it seems difficult to use a similar approach for rectangles because it is not clear how to apply a cascading procedure efficiently. 
Indeed, an important difference with the previous cases is that the $x$ and $y$ coordinates can behave independently, and thus it is hard to track the relationship between the two dimensions with a cascading procedure 
 (see Figure~\ref{fig:QuadraticRectangles}).

Instead, we use a different approach to compute all stabbing rectangles. Any inclusionwise smallest stabbing
rectangle $R$ must contain one endpoint
on each side, and in particular, an endpoint $v$ of $S$ must be on
its lower boundary segment (otherwise, we could shrink it further).

The key observation is that if we fix $v$ (or equivalently, the lower side of a candidate stabbing rectangle), then we can reduce the problem to that of finding a $3$-sided rectangle.
In particular, by fixing the lower side of the rectangle we are forcing all points below it to be blue (and as a result the corresponding endpoint of the segment must be red).
After a successful cascading procedure, we end up with a certain initial classification that can be completed to a solution to the rectangle if and only if a compatible stabbing $3$-sided rectangle exists. Since there are $2n$ candidates for $v$ (each of the endpoints of the segments of $S$), we have $O(n)$ different instances that can be solved independently using Theorem~\ref{theo_bottomless}.
Therefore we obtain the following result.

\begin{theorem}\label{theo_rect}
All the combinatorially different stabbing rectangles for a
set $S$ of $n$ segments can be computed in $O(n^2 \log n)$ time.
\end{theorem}
\newcommand{\pftheorect}{
\begin{proof}
The bottleneck of the algorithm is the $O(n)$ times that we invoke the algorithm of Theorem~\ref{theo_bottomless}. Between two such calls, we  execute a cascading procedure. As in previous cases, a segment can only be classified once in this manner, thus this cascading will need overall $O(n\log n)$ time.

Each time we find a stabbing rectangle, we must make sure it is non-trivial (i.e., determine if it can be extended to a stabbing 3-rectangle). This is done in $O(\log n)$ time with a constant number of orthogonal range search queries (as described in the proofs of Theorems~\ref{theo_allquadrants} and~\ref{theo_bottomless}). Also note that the linear bound on the number of possible 3-rectangles (Lemma~\ref{lem_atmostn}) immediately implies a quadratic bound on the number of possible stabbing rectangles. In fact, it is not hard to build a set of segments with $\Theta(n^2)$ combinatorially different stabbing rectangles (such an example is shown in Figure~\ref{fig:QuadraticRectangles}).
\end{proof}
}
\ShoLong{}{\pftheorect}

\begin{figure}[tb]
\begin{center}
\includegraphics{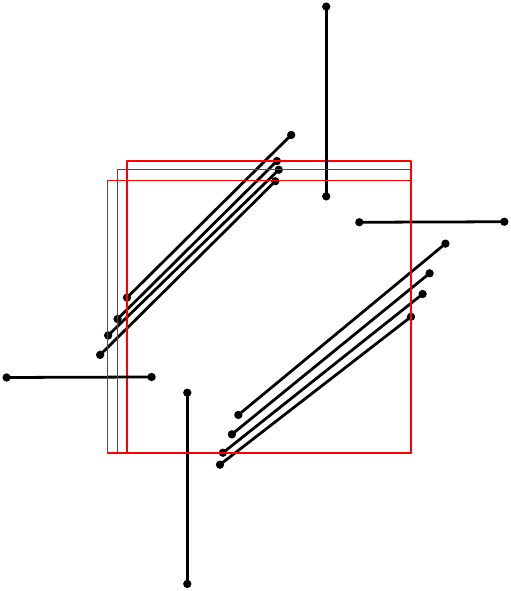}
\end{center}
\caption{A set of segments that has $\Theta(n^2)$ combinatorially different stabbing rectangles.
In the example, the choice of the bottom-right corner is independent of the choice the top-left corner, leading to a quadratic number of rectangular stabbers.}
\label{fig:QuadraticRectangles}
\end{figure}

\section{Conclusions and open problems}\label{sec_concl}

In this paper we introduced a general algorithm for computing all combinatorially different stabbers of a set of segments, which
classifies the endpoints based on partitioning the
plane into red, blue, and gray regions that are updated as points
become classified.
We showed how to apply such strategy
to horizontal strips, quadrants, $3$-sided rectangles
and, indirectly, to axis-aligned rectangular stabbers.
Furthermore, we have proved that our algorithms for horizontal strips, quadrants, and $3$-sided rectangles are asymptotically optimal (under a mild assumption on the description of the stabbing region), while the algorithm for rectangles has only a logarithmic factor overhead.

\begin{figure}[tb]
\begin{center}
\includegraphics{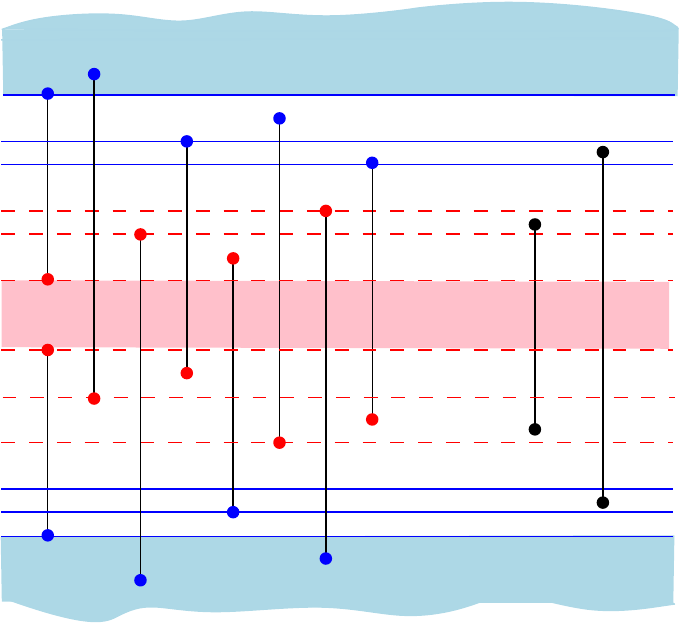}
\end{center}
\caption{A set of segments that, when processed from left to right, creates a series of growing nested red and blue regions. The existence of a solution depends on whether at some point the growing red or blue regions contain both endpoints of some segment, like the two rightmost ones.}
\label{fig:encadena}
\end{figure}

This work gives rise to several directions for future work.
One of them is studying the decision version of the problems.
That is, what if one only needs to know if one stabber of a certain shape exists?
It is an intriguing open problem whether one can solve the decision versions of the problems faster, even for the case of horizontal strips.
This was shown not to be possible for line stabbers~\cite{ARW,emprww}, but it is unclear if this is also the case in our context.
In our setting, an important difficulty for devising a more efficient algorithm for the decision version is that there does not seem to exist any local condition that determines whether a solution exists.
As shown in Figure~\ref{fig:encadena}, it may be necessary to analyze a linear number of segments (with the corresponding cascading iterations) to find out that a subset of them forces a stabbing region that is not compatible with some other segment.
It is unclear whether some different approach could yield a $o(n \log n)$ algorithm for this problem.

Other natural algorithmic extensions include designing an output-sensitive algorithm to report all combinatorially different stabbing rectangles, with running time proportional to the number of solutions, and studying the optimization variants of the problems for the cases in which no stabber exists. For those cases, it would be useful to have an efficient algorithm to find a stabber that classifies as many segments as possible.

Finally, another relevant family of problems arises from looking for stabbers with arbitrary orientations (instead of rectilinear).
For the case of lines, going from fixed-orientation lines to any line increases the running time from $O(n)$ to $\Omega(n \log n)$ time~\cite{emprww}. In our setting, there are only $O(n^2)$ relevant directions, hence our algorithms could be used for this variants if applied to each of the relevant directions, yielding an $O(n^2)$ multiplicative increase in the running time.
Designing more efficient algorithms for this case is clearly an interesting open problem.

%

\paragraph{Acknowledgments.} {\small
M. C., C. S., and R.I. S were partially supported by projects MINECO MTM2015-63791-R/FEDER and Gen. Cat. DGR2014SGR46. D. G. was
supported by projects PAI FQM-0164 and MTM2014-60127-P. 
 M.~K.~was supported in part by the ELC project (MEXT KAKENHI No.~12H00855 and 15H02665).
 R.I.~S. was also partially supported by MINECO through the Ram{\'o}n y Cajal program.}


\begin{thebibliography}{10}

\bibitem{AHIKLMPS01}
M.~Abellanas, F.~Hurtado, C.~Icking, R.~Klein, E.~Langetepe, L.~Ma, B.~Palop, V.~Sacrist\'{a}n.
\newblock Smallest color-spanning objects.
\newblock In: Meyer auf der Heide, F. (ed.) ESA 2001. LNCS, vol. 2161. Springer, Heidelberg, (2001),  pp. 278--289.

\bibitem{ahikl-fcvd-06}
M.~Abellanas, F.~Hurtado, C.~Icking, R.~Klein, E.~Langetepe, L.~Ma, B.~Palop, V.~Sacrist\'{a}n.
\newblock The farthest color {V}oronoi diagram and related problems.
\newblock Technical Report 002, Rheinische--Friedrich--Wilhelms Universit\"{a}t Bonn, 2006.


\bibitem{aegmps-dlcdc-11}
P.K. Agarwal, A. Efrat, C. Gniady, J.S.B. Mitchell, V. Polishchuk, G. Sabhnani.
\newblock Distributed localization and clustering using data correlation and the Occam's razor principle.
\newblock In \emph{Proc. 7th International Conference on Distributed Computing in Sensor Systems}, (2011).

\bibitem{abcckm-cfc-15}
E.M. Arkin, A. Banik, P. Carmi, G. Citovsky, M.J. Katz, J.S.B. Mitchell, M. Simakov.
\newblock Conflict-free Covering.
\newblock In \emph{Proc. 27th Canadian Conference on Computational Geometry}, (2015).

\bibitem{adkmpsy-ct-14}
E.M. Arkin, C. Dieckmann, C. Knauer, J.S.B. Mitchell, V. Polishchuk, L. Schlipf, S, Yang.
\newblock Convex transversals.
\newblock \emph{Computational Geometry: Theory and Applications}, Vol. 47, (2014), pp. 224–-239.

\bibitem{aghmz-mlscp-15}
E.M. Arkin, J. Gao, A. Hesterberg, J.S.B. Mitchell, J. Zeng.
\newblock The Minimum Length Separating Cycle Problem.
\newblock In \emph{Abstracts 25th Fall Workshop on Computational Geometry}, (2015).

\bibitem{AB}{M. Atallah, C. Bajaj. Efficient algorithms for common
transversal. \emph{Information Processing Letters}, Vol. 25, (1987),
pp. 87--91.}

\bibitem{ARW}{D. Avis, J.-M. Robert, R. Wenger. Lower bounds for line
stabbing. \emph{Information Processing Letters}, Vol. 33, (1989), pp. 59--62.}

\bibitem{aw88}{D. Avis, R. Wenger. Polyhedral line transversals in space.
\emph{Discrete and Computational Geometry}, Vol. 3, (1988), pp. 257--265.}

\bibitem{similar}{L. Barba, S. Durocher, R. Fraser, F. Hurtado, S. Mehrabi, D. Mondal, J. Morrison, M.  Skala, M. A. Wahid.
On $k$-enclosing objects in a coloured point set. \emph{Proc. 25th Canadian
Conference on Computational Geometry}, (2013), pp. 229--234.}

\bibitem{BCE}{B.K. Bhattacharya, J. Czyzowicz, P. Egyed,
I. Stojmenovic, G. Toussaint, J. Urrutia. Computing shortest transversals of
sets. 
\emph{International Journal of Computational Geometry and Applications}, Vol. 2(4), (1992), pp. 417--442.}

\bibitem{BKM}{B. Bhattacharya, C. Kumar, A. Mukhopadhyay.
Computing an area-optimal convex polygonal stabber of a set of
parallel line segments. \emph{Proc. 5th Canadian
Conference on Computational Geometry}, (1993), pp. 169--174.}

\bibitem{belzw05}{H. Br\"{o}nnimann, H. Everett, S. Lazard, F. Sottile, S.
Whitesides. Transversals to line segments in three-dimensional
space. \emph{Discrete and Computational Geometry}, Vol. 34, (2005), pp. 381--390.}

\bibitem{mc}{M. Claverol. Problemas geom\'etricos en morfolog{\'\i}a
computacional. \emph{PhD Thesis}. Universitat Polit{\`e}cnica de Catalunya}, 2004.


\bibitem{cggms}{M. Claverol, D. Garijo, C. I. Grima, A. M\'arquez, C. Seara. Stabbers of line segments in the plane.
\emph{Computational Geometry: Theory and Applications}, Vol. 44(5), (2011), pp. 303--318.}

\bibitem{fct}{M. Claverol, D. Garijo, M. Korman, C. Seara, R. I. Silveira.
Stabbing Segments with Rectilinear Objects.
In: Kosowski, A., Walukiewicz, I. (eds.)
FCT 2015. LNCS, vol. 9210, pp. 53--64. Springer, Heidelberg (2015).
}

\bibitem{CKPSS15}{M. Claverol, E. Khramtcova, E.  Papadopoulou, M. Saumell, C. Seara.
Stabbing circles for sets of segments in the plane.
In \emph{Abstracts XVI Spanish meeting on Computational Geometry}, (2015), pp. 112--115.}


\bibitem{DGN} {S. Das, P. P. Goswami, S. C. Nandy.
Smallest color-spanning objects revisited.
\emph{International Journal of Computational Geometry and Applications}, Vol. 19(5), (2009), pp. 457--478.}

\bibitem{DKPPSS}{J. M. D\'{i}az-B\'a\~{n}ez, M. Korman, P. P\'erez-Lantero,
A. Pilz, C. Seara, R. I. Silveira. New results on stabbing
segments with a polygon.
{\em Computational Geometry: Theory and Applications}, Vol. 48(1), (2015), pp. 14-–29.}

\bibitem{DL}{O. Daescu, J. Luo. Stabbing balls and simplifying proteins. 
\emph{International Journal of Bioinformatics Research and Applications}, Vol. 5(1), (2009), pp. 64--80.}


\bibitem{emprww}{H. Edelsbrunner, H. A. Maurer, F. P. Preparata,
A. L. Rosenberg, E. Welzl, D. Wood. Stabbing line segments.
\emph{BIT}, Vol. 22, (1982) pp. 274--281.}

\bibitem{fhph12}{E. Fogel, M. Hemmer, A. Porat, D. Halperin.
Lines through segments in three dimensional space.
\emph{Proc. 29th European Workshop on Computational
Geometry}, (2012), pp. 113--116.}

\bibitem{ghms-apsmlp-93} {L. J. Guibas, J. Hershberger, J. S. B. Mitchell, J. Snoeyink:.
Approximating Polygons and Subdivisions with Minimum Link Paths.
\emph{International Journal of Computational Geometry and Applications}, Vol. 3(4), (1993), pp. 383--415 .}

\bibitem{gpw-gtt-93}{J. E. Goodman, R. Pollack, R. Wenger.
\newblock Geometric Transversal Theory.
\newblock In \emph{New Trends in Discrete and Computational Geometry}, J. Pach, ed., (1993), Springer, Berlin, pp. 163--198.}

\bibitem{GS}{M. T. Goodrich, J. S. Snoeyink.
Stabbing parallel segments whith a convex polygon.
\emph{Proc. 1st Workshop Algorithms and Data
Structures}, (1989), pp. 231--242.}




\bibitem{KK}{A. Kaneko and M. Kano.
Discrete geometry on red and blue points in the plane --- A survey.
\emph{Discrete and Computational Geometry
Algorithms Combin.},
Vol. 25, (2003), pp. 551--570.}


\bibitem{krs10}{H. Kaplan, N. Rubin, M. Sharir. Line transversal of convex
polyhedra in $\mathbb{R}^3$. \emph{SIAM Journal on Computing},
Vol. 39(7), (2010), pp. 3283-3310.}


\bibitem{mg-06}{D. Merrick, J. Gudmundsson.
Path Simplification for Metro Map Layout.
In: Kaufmann, M., Wagner, D. (eds.)
GD 2016. LNCS, vol. 4372, pp. 258--269. Springer, Heidelberg (2006).
}

\bibitem{McC85} {E. M. McCreight. Priority search trees.
\emph{SIAM Journal on Computing}, Vol. 14, (1985), pp. 257--276.}

\bibitem{MKGB}{A. Mukhopadhyay, C. Kumar, E. Greene, B. Bhattacharya.
On intersecting a set of parallel line segments with a convex
polygon of minimum area. \emph{Information Processing Letters},
Vol. 105, (2008), pp. 58--64.}

\bibitem{MGR}{A. Mukhopadhyay, E. Greene, S. V. Rao.
On intersecting a set of isothetic line segments with a convex
polygon of minimum area. \emph{International Journal of Computational
Geometry and Applications}, Vol. 19(6), (2009), pp. 557--577.}

\bibitem{OR1}{J. O'Rourke. An on-line algorithm for fitting straight lines between data ranges.
\emph{Communications of the ACM}, Vol. 24, (1981), pp. 574--578.}

\bibitem{p93}{M. Pellegrini. Lower bounds on stabbing lines in 3-space.
\emph{Computational Geometry: Theory and Applications}, Vol. 3, (1993),
pp. 53--58.}

\bibitem{R}{D. Rappaport. Minimum polygon transversals of line
segments. \emph{International Journal of Computational Geometry
and Applications}, Vol. 5(3), (1995), pp. 243--256.}

\end{thebibliography}
\end{document}